\documentclass[aps]{revtex4}
\usepackage{amsmath, amssymb, amsthm, hyperref, braket, enumitem, color}
\usepackage[normalem]{ulem}
\usepackage[all]{xy}
\CompileMatrices

\newtheorem{de}{Definition}
\newtheorem{ex}{Example}
\newtheorem{lemma}{Lemma}
\newtheorem{theorem}{Theorem}
\newtheorem*{theorem*}{Theorem}
\newtheorem{remark}{Remark}

\DeclareMathOperator{\tr}{Trace}

\begin{document}
\title{Bipartite separability and non-local quantum operations on graphs}
\author{Supriyo Dutta}
\affiliation{Department of Mathematics, Indian Institute of Technology Jodhpur, Jodhpur 342011, India}
\email{dutta.1@iitj.ac.in}
\author{Bibhas Adhikari}
\affiliation{Department of Mathematics, Indian Institute of Technology Kharagpur, Kharagpur 721302, India}
\email{bibhas@maths.iitkgp.ernet.in}
\author{Subhashish Banerjee}
\affiliation{Department of Physics Indian Institute of Technology Jodhpur, Jodhpur 342011, India}
\email{subhashish@iitj.ac.in}
\author{R. Srikanth}
\affiliation{Poornaprajna Institute of Scientific Research, Bangalore, Karnataka 560080, India}
\email{srik@poornaprajna.org}

\begin{abstract}
\noindent 
\textbf{Abstract}
In this paper, we consider the separability problem for bipartite quantum states arising from graphs. Earlier it was proved that the degree criterion is the graph theoretical counterpart of the familiar PPT criterion for separability, although there are entangled states with positive partial transpose for which degree criterion fails. Here, we introduce the concept of partially symmetric graphs and degree symmetric graphs by using the well-known concept of partial transposition of a graph and degree criteria, respectively. Thus, we provide classes of bipartite separable states of dimension $m \times n$ arising from partially symmetric graphs. We identify partially asymmetric graphs which lack the property of partial symmetry. Finally we develop a combinatorial procedure to create a partially asymmetric graph from a given partially symmetric graph. We show that this combinatorial operation can act as an entanglement generator for mixed states arising from partially symmetric graphs.

\end{abstract}

\maketitle

\section{Introduction}

Graph   theory   \cite{west2001,  bapat2010}   is   a
  well-established branch of mathematics. It forms the core of complex
systems \cite{northrop2010introduction, chung2006complex},  widely used in  Economics, Social
  Science  and  System  Biology   \cite{lu2011link},  as  well  as  in
  communication and  information \cite{han2012graph}. It is  also used
  to address foundational aspects of different branches of mathematics
  and physics \cite{bianconi2015interdisciplinary}. To the best of our
  knowledge, combinatorial graphs have  been used in quantum mechanics
  and information  theory \cite{nielsen2010quantum} in  four different
  ways: 
  (a). {\bf  Quantum Graphs:}  Here, a differential or pseudo-differential operator is associated with a graph. The operator acts on functions defined on each edge of the graph when the edges are equipped with compact real intervals. \cite{berkolaiko2013introduction, berkolaiko2006quantum}. 
  (b). {\bf Graph states:} In this approach, combinatorial graphs are used to
describe interactions   between  different  quantum   states  \cite{hein2004multiparty,
    anders2006fast,  benjamin2006brokered}. Here  the vertices  of the
  graph  represent  the  quantum mechanical  states,
  while the  interactions between them  are represented by  the edges.
  Graph states  were proposed as  a generalization of  cluster states,
  which  is   the  entanglement  resource  used   in  one-way  quantum
  computation.
(c). \textbf{Combinatorial approach  to LOCC  (local operations
  and classical communication) transformations in multipartite quantum
  states:} Here graph  theoretic methods were applied  to the analysis
of pure maximally entangled  quantum states distributed among multiple
geographically separated parties \cite{SPS+05,PSS06}.
(d). {\bf Braunstein, Ghosh and Severini  approach:} Here, a single quantum
state    is    represented    by   a    graph    \cite{braunstein2006,
  adhikari2012}.  Combinatorial  properties  of a  quantum  mechanical
state can be studied using this approach.

This  work is  in the  spirit of  the Braunstein  et.  al.   approach.
Representing a quantum state by a  graph is beneficial for research in
both quantum information  theory as well as  complex networks.  Graphs
provide   a  platform   to   visualize   quantum  states   pictorially
\cite{ionicioiu2012encoding},   such   that  different   states   have
different pictographic        representations
\cite{adhikari2012} and some important  unitary evolutions can also be
represented by changes in  their representations \cite{dutta2016}.  In
this  way,   graphs  form  an intuitively  appealing
  framework for quantum information  and communication.  On the other
hand, measuring entropy and complexity of large, complex networks is a
challenging part  of network science. Correspondence between graphs
and  quantum  states  provide  an insightful  connection  between  the
Shannon and von-Neumann entropy on the one hand and, on the other, the
complexity   of   networks  \cite{du2010note,   passerini2008neumann},
details   of   which   can   be   seen   from   \cite{zhao2011entropy,
  anand2009entropy, anand2011shannon, maletic2012combinatorial}.  This
interconnection has also been exploited in quantum gravity and quantum
spin   networks   \cite{rovelli2010single}.

A combinatorial  graph  $G =  (V(G),  E(G))$ is  an
  ordered pair of  sets $V(G)$ and $E(G),$ where $V(G)$  is called the
  vertex set and  $E(G)\subseteq V(G)\times V(G)$ is the  edge set. In
  this paper,  we are concerned  with simple graphs, which  are graphs
  without multiple edges and loops. Between any two vertices there is
  a  maximum  of one  edge.  There  is no  edge  linking  a vertex  to
  itself. An edge is denoted by  $(i, j)$ which links the vertices $i$
  and $j$. The  adjacency matrix $A(G) = (a_{ij})$  associated with a
  simple  graph $G$  is a  binary (all  elements are  $0,1$) symmetric
  matrix defined as
$$a_{ij} = \begin{cases}
1 & \text{if}~ (i,j) \in E(G),\\
0 & \mbox{otherwise}.
\end{cases}$$
Thus,  the order  of $A(G)$  is $|V(G)|$  where, $|V(G)|$  denotes the
number of elements  of the vertex set, $V(G)$. The  degree of a vertex
$u$ is  the number of edges  incident to it, denoted  by $d_G(u)$. The
degree  matrix  $D(G)$  of  $G$   is  the  diagonal  matrix  of  order
$|V(G)|$. Its $i$-th diagonal entry is  the degree of the $i$th vertex
of $G$, $i =  1, 2, \dots, |V|$. Two simple graphs  $G_1$ and $G_2$ are
isomorphic if  there exists  a bijective map  $f :  V(G_1) \rightarrow
V(G_2)$, such that $(i,j) \in E(G_1)$  if and only if $(f(i), f(j))\in
E(G_2)$. When  $G_1$ and $G_2$  are isomorphic there is  a permutation
matrix $P$ such that $A(G_1) = P^TA(G_2)P$.

In  quanutm mechanics  a density  matrix $\rho$  is a
  positive    semidefinite    Hermitian unit-trace
  matrix. Familiar  positive semidefinite matrices related  to a graph
  are  the  combinatorial  Laplacian  matrix  $L(G)  =  D(G)  -  A(G)$
  \cite{bapat2010}, the signless Laplacian matrix $Q(G) = D(G) + A(G)$
  \cite{cvetkovic}   and  the   normalised  Laplacian   matrix  $M(G)$
  \cite{banerjee2007spectrum,  wu2016graphs}.  In  this work,  we  are
  concerned  with the  density  matrices corresponding  to $L(G)$  and
  $Q(G)$ only. They are defined as \cite{adhikari2012}
$$\rho_l(G) = \frac{L(G)}{\tr(L(G))} ~\mbox{and}~ \rho_q(G) = \frac{Q(G)}{\tr(Q(G))}.$$
For any two isomorphic graphs $G_1$ and $G_2$,
$$L(G_1) = P^TL(G_2)P \hspace{.5cm} ~\mbox{and}~ \hspace{.5cm} Q(G_1) = P^TQ(G_2)P.$$
$$\Rightarrow    \rho_l(G_1)     =    P^T\rho_l(G_2)P    \hspace{.5cm}
~\mbox{and}~ \hspace{.5cm} \rho_q(G_1) = P^T\rho_q(G_2)P.$$ Throughout
this paper we shall denote a  general density matrix by  $\rho$, while
$\rho_l(G)$ and  $\rho_q(G)$ are specific density  matrices as defined
above, collectively written as $\rho(G)$.

Here,  we  are  concerned  with  bipartite
  systems distributed  between two parties  $A$ and $B$. It is well known that a state of such a system, represented by the density matrix $\rho$, is
  separable if and only if  it can be represented as a
    convex combination of product states, i.e., there are two sets of
  density  matrices  $\{\rho_k^{(A)} :  ~\mbox{order}(\rho_k^{(A)})  =
  m\}$   and  $\{\rho_k^{(B)}:   ~\mbox{order}(\rho_k^{(B)})  =   n\}$
  corresponding to $A$ and $B$ respectively, such that,
$$\rho = \sum_k p_k \rho_k^{(A)} \otimes \rho_k^{(B)}; \sum_k p_k = 1,
  ~ p_k \ge  0.$$ Here and below, $\otimes$ denotes  tensor product of
  matrices \cite{hom}.   Trivially, the  dimension of $\rho$  is $mn$. The state corresponding to
  $\rho$   is    called   entangled    if   it   is    not   separable
  \cite{horodecki2009quantum}.   If $k  =  1$ in  the above  equation,
  $\rho$ is called a pure  state.  Else, it is a mixed state which is  a probabilistic  mixture of
  different pure  states. Detection of entangled  states, known as
the ``quantum separability problem'' (QSP),  is one of the fundamental
problems of the quantum information theory \cite{guhne2009} due to its
wide  applications in  various quantum  communication and  information
processing  tasks. The Peres-Horodecki  criterion
  \cite{peres1996,horodecki1997,mcmahon2007quantum}, also known as the
  positive  partial transpose  (PPT) criterion,  provides a  necessary
  condition for separability.  It also provides a sufficient condition
  for  systems  of  dimension  $2\times2$  and  $2\times3$.   However,
  sufficiency for higher dimensional systems requires in general other
  techniques,  like  entanglement  witness.  As  $\rho$  is  a
  matrix of  order $mn$, it  can be written as  an $m \times  m$ block
  matrix with each block of size  $n \times n$.  The partial transpose
  corresponding to $B$, denoted by $\rho^{T_B}$, is obtained by taking
  individual transpose  of each block  \cite{mcmahon2007quantum}.  The
  PPT criterion states that for any separable state, $\rho^{T_B}$ is a
  positive  semi-definite   matrix  \cite{peres1996}.    However,  the
  converse is true only for bipartite systems of dimensions $2\times 2$
  and $2\times 3$ \cite{horodecki1997}.  There  are a number of other
separability critera \cite{horodecki2009quantum}.

The graph  theoretic approach to  solving QSP  has generated a  lot of
interest   in    the   last    decade   after   the    seminal   paper
\cite{braunstein2006}. This approach is beneficial as
  it is more  efficient for mixed states. The state  $\rho(G)$ is pure
  if  it   consists  of   a  single  edge;   otherwise  it   is  mixed
  \cite{braunstein2006, adhikari2012}.  The  separability of bipartite
  quantum  states  corrsponding to  random  graphs  are considered  in
  \cite{garnerone2012bipartite}.    Some   families  of   graphs
    were invented for  which separability
  can  be  tested  easily   \cite{braunstein2006some}.   The  idea  of
  entangled edges  \cite[Section 4.3]{braunstein2006}  was generalised
  in   \cite{rahiminia2008separability}. Motivated  by   the   PPT
criteria,  the  QSP   problem  for $\rho_l(G)$ was
considered  in  \cite{severini2008},  where, the  concept  of  partial
transpose   was   introduced  graph theoretically.
It introduced the degree criteria as  the condition
for separability.  However, the degree criteria failed to detect bound
entangled          states, that          is,
entangled  states  with positive  partial  transpose.
Thus,  finding  sufficient  conditions  on graphs  that  can  generate
separable states is  a current topic of interest in  the literature. A
class of  graphs which produce  $2 \times p$ separable  quantum states
were   identified   in  \cite{wu2006}. The   degree
  criterion    was    generalised    for    tripartite    states    in
  \cite{wang2007tripartite}.        In      \cite{xie2013separability,
    wu2009multipartite,  wu2010graphs}  QSP   for  higher  dimensional
  states  were addressed.  For some  particular class  of graphs,  the
  properties  of  corresponding  quantum   states  were  discussed  in
  \cite{hui2013separability,  li2015quantum}.   An  interesting  fact,
  already  discussed   in  the  literature  regarding   QSP,  is  that
  separability of $\rho(G)$ does not  depend on graph isomorphism. Two
  isomorphic graphs  may correspond  to quantum states  with different
  separability      properties      \cite{braunstein2006,severini2008,
    wu2010graphs}.   This  is  contradictory to  our  classical  world
  phenomena,  wherein  any  two  isomorphic graphs  possess  the  same
  properties.

In  \cite{severini2008},  the  degree  criterion  was
  shown  to be  equivalent to  the PPT  criterion.  Hence,  a stronger
  criterion for separability than  the degree criterion is essential.
Inspired  by the  degree criteria,  in  this paper,  we define  degree
symmetric  graphs. The  motivation for  this is  that
  entanglement  of  $\rho_l(G)$ and  $\rho_q(G)$  may  depend on  some
  symmetry hidden in the graph. Inspired by this idea we define here a
  notion   of  partial   symmetry.   We   generalise  the   result  of
  \cite{wu2006} to partially symmetric graphs. Then we derive a class
of partially  symmetric graphs which produce  separable quantum states
 $\rho_l(G)$ and  $\rho_q(G)$ of dimension  $m \times
n$. To the  best of our knowledge, there are  no sufficient conditions
till date for separability of $m\times n$ systems arising from graphs.
How  to generate  bigger  graphs providing  separable
  states from smaller graphs? We  define a graph product $G\bowtie H$
for  a simple  graph $G$  and a  partially symmetric  graph $H$  which
corresponds to separable bipartite states.

We  collect our  ideas  related  to separability  and
  partially symmetric graphs in the Section 2.  Here we also introduce
  the concept  of multi-layered system  in the context of  graphs.  In
  Section 3, we  use graph isomorphism as  an entanglement generator.
As a  by-product of the  separability criteria, we propose  some graph
isomorphisms, which are non-local  in nature, to generate entanglement
from a given partially symmetric graph. Finally, we provide an example
of an entangled state generated  by employing this non-local operation
on a partially symmetric graph which represents separable states. Thus,
we  conclude  from this  example  that  non-local operations  are  not
limited to the use of CNOT  gate operations on separable states, as is
observed in quantum information theory.   We then make our conclusions
and bring out some open problems arising from this work.

\section{Partial symmetric graphs and separability}

This section begins with the  creation of layers in a
  graph $G$. It partitions density  matrices $\rho(G)$ into blocks. We
  also define  graph theoretical  partial transpose (GTPT),  the graph
  theoretical analogue  of partial  transpose. This is  an equivalence
  relation on the  set of all graphs. GTPT  equivalent graphs preserve
  the  separability  property.   Next,  we  define  partial  symmetric
  graphs. A sufficiency condition is provided for separability of states which arise from
    partially  symmetric graphs.   We also  define  a
  product operation for two graphs such that the density matrices corresponding to the resultant graph represent separable states.

Let the vertex set of the graph $G$, $V(G)$, with $mn$ number of vertices be labelled by integers $1, 2, \dots, mn$. Then partition $V(G)$ into $m$ layers with $n$ vertices in each layer. Let the layers be $C_1,  C_2, \dots , C_m$ where $C_i = \{v_{i,1}, v_{i,2}, \dots v_{i,n}\}$ and $v_{i,k} = ni + k$. This allows $A(G)$ to be partitioned into blocks as follows.
\begin{equation}
A(G) = \begin{bmatrix}
		A_{1} & A_{1,2} & \dots & A_{1,m} \\
		A_{2,1} & A_{2} & \dots & A_{2,m} \\
		\vdots & \vdots & \vdots & \vdots\\
		A_{m,1} & A_{m,2} &\dots &A_{m}
		\end{bmatrix},
\end{equation}
where $A_{i,j}, i \neq j$  and $A_i$ are matrices of order  $n$. $A_{i,j}, i\neq j$ represents edges between $C_i$  and $C_j$.  $A_i$ represents edges between vertices of $C_i$. Trivially, $A_{i}^T = A_{i}$ and $A_{i,j}^T = A_{j,i}$ for all $i\neq j$. Observe that $A_{i,j}$ need not be symmetric. Throughout this article, $G$ is a simple graph with standard labelling on the vertex set $V(G) = \{1, 2, \dots ,mn\}$ with layers as described above.

This  article  deals with quantum  entanglement of bipartite states of dimension $m\times n$ which arise from simple graphs of $mn$ vertices. We mention that the  bipartition does not  exist a-priori in the graph, but is induced  by the above layering.  We wish  to understand how the two abstract ``particles'', created by this induction  based on vertex labellings are related to $G$.  $V(G)$ is arranged as a matrix of dots as follows.
\begin{align}
C_1 &= \xymatrix{\bullet_{v_{1,1}} & \bullet_{v_{1,2}} & \dots & \bullet_{v_{1,n}}}\nonumber \\
C_2 &= \xymatrix{\bullet_{v_{2,1}} & \bullet_{v_{2,2}} & \dots & \bullet_{v_{2,n}}}\nonumber \\
\vdots& \nonumber \\
C_m &= \xymatrix{\bullet_{v_{m,1}} & \bullet_{v_{m,2}} & \dots & \bullet_{v_{m,n}}}
\label{al:particle}
\end{align}
Effectively, one particle, of dimension  $n$, is assumed to correspond to horizontal  direction, whilst  another particle, of  dimension $m$, corresponds to the  perpendicular direction.  Thus, the first and second indices of every vertex label $A_{j,k}$ comes from the vertical and horizontal particles, respectively.  More particles can be induced in the system in different orthogonal  directions by drawing $G$ in an orthogonal  higher  dimensional  structure,  which  will  be  explored elsewhere.  In  the analogous construction  for a 3-partite  system of dimension  $lmn$,  we can arrange  the  entries  of $A_{j,k}$  as  a three-dimensional stack, with the vertical  layer of height $l$. Then the  entries $A_{j,k}$ ($j=1,\cdots,m; k=1,\cdots,n$)  will be on the ``ground'' layer, with the next layer having the entries $A_{j,k}$ $(j=m+1,\cdots,2m; k=n+1,\cdots,2n)$  and in  general the  $r$th layer ($1\le r\le  l$) having the entries  $A_{j,k}$ ($j=(r-1)m+1,\cdots,rm; k=(r-1)n+1,\cdots,rn$).  Note  that this scheme can  be introduced for any number of  induced particles, but the  simple assignment of direction to  particles  as  ``vertical''  and  ``horizontal''  will  no  longer be possible for three or more particles.

Let us return to  the bipartite case. As defined in \cite{severini2008, wu2006}, we recall that partially transposed graph $G'$ is obtained by employing the algebraic partial transposition to the adjacency matrix of a given graph $G.$ This idea is equivalent to partial transpose on the 2-nd party in a bipartite systems density matrix. For convenience in dealing with our labelling of the vertices in the graph $G$, we reformulate the definition of partially transposed graph by introducing it as a by-product of the following combinatorial operation.

\begin{de}
Graph theoretical partial transpose (GTPT) is an operation on the graph $G$ replacing  all existing edges $(v_{i,k},  v_{j,l}),  k  \neq  l,  i \neq  j$  by $(v_{i,l},v_{j,k})$, keeping all other edges unchanged.
\end{de}

Thus, GTPT generates a new simple graph $G^\prime = (V(G^\prime),E(G^\prime))$ from a given simple graph $G=(V(G), E(G))$ where $V(G^\prime) = V(G)$ with the labelling unchanged. Note that,  $G$ can also  be constructed from $G^\prime$ by GTPT as, $(G^\prime)^\prime = G$.  We  call $G$  and $G^\prime$ as  GTPT equivalent. It is easy to verify that $A(G^\prime) = A(G)^{T_B}$ and hence $|E(G)|=|E(G')|.$

\begin{ex}\label{K13}
The GTPT of the star graph with four vertices is depicted below.
$$
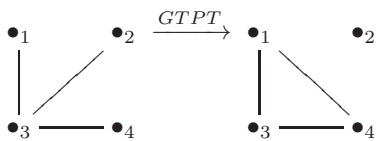
\xymatrix{\bullet_1 \ar@{-}[d] & \bullet_2\ar@{-}[dl] \\ \bullet_3 \ar@{-}[r]& \bullet_4} \xrightarrow{GTPT} \xymatrix{\bullet_1 \ar@{-}[d] \ar@{-}[dr]& \bullet_2 \\ \bullet_3 \ar@{-}[r]& \bullet_4}$$
\end{ex}

The Example \ref{K13} establishes that GTPT of a connected graph need not be connected. Also, it changes the degree sequence of the graph. A relevant question here is: does there exist a graph for which the degree sequence remains invariant under GTPT? Inspired by the degree criteria introduced in \cite{braunstein2006, severini2008}, we define degree symmetric graphs as follows.

\begin{de}
A graph $G$ is called degree symmetric if $d_G(u) = d_{G^\prime}(u) \,\, \mbox{for all} \,\, u\in V(G) = V(G^\prime)$.
\end{de}

Thus, for a degree symmetric graph, the degree sequence of the graph is preserved under GTPT. The following is an example of a degree symmetric graph.
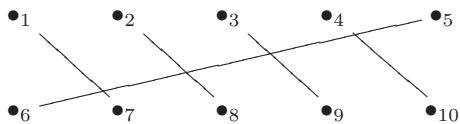
\begin{ex}\label{degree symmetric graph}\cite{braunstein2006}
$$\xymatrix{\bullet_1 \ar@{-}[dr] & \bullet_2 \ar@{-}[dr] & \bullet_3 \ar@{-}[dr] & \bullet_4 \ar@{-}[dr] & \bullet_5 \ar@{-}[dllll] \\ \bullet_6 & \bullet_7 & \bullet_8 & \bullet_9 & \bullet_{10}}$$
\end{ex}

It was conjectured in \cite{braunstein2006} that $\rho_l(G)$ is a separable bipartite state in any dimension if and only if $G$ and $G'$ have the same degree sequence. In other words, $\rho_l(G)$ is separable if and only if $G$ is a degree symmetric graph. Later the conjecture was proved to be false in \cite{severini2008}. An example of a degree symmetric graph $G$ was provided for which $\rho_l(G)$ is entangled. It was established that PPT criteria is equivalent to the degree criteria for $\rho_l(G)$. However, the separability of $\rho_l(G')$ and $\rho_q(G')$ was not discussed there \cite{severini2008}. In this work we prove that degree symmetric graphs preserve the separability even after GTPT. This result can be stated as a theorem.

\begin{theorem}\label{theorem1}
Separability    of   $\rho_l(G)$    implies   the    separability   of
$\rho_l(G^\prime)$ if and  only if $G$ is  degree symmetric. Similarly
separability of  $\rho_q(G)$ implies  separability of  $\rho_q(G')$ if
and only if $G$ is degree symmetric.
\end{theorem}
The proof can be found in the appendix.

Now we introduce the concept  of partially symmetric graphs. This will
play  a  central   role  in  the  development  of  the   rest  of  the
paper. Our aim is to make a more stringent condition
of symmetry in a degree symmetric graph. We focus on
symmetry in  the partial  transposition of the  adjacency matrix  of a
graph and  hence, define partially  symmetric graphs (in  analogy with
``partial transpose'') as follows.

\begin{de}
A graph $G$ is partially symmetric if $(v_{i,l}, v_{j,k}) \in E(G)$ implies $(v_{i,k}, v_{j,l}) \in E(G) ~\forall~ i, j, k, l, i\neq j$.
\end{de}

In the above definition, $i$ and $j$ indicates layers $C_i$ and $C_j$ such that vertices $v_{i,l} \in C_i$ and  $v_{j,k} \in C_j$. Suffixes $l$ and $k$ represents the relative positions of the vertices in the individual layers.

Note that, GTPT keeps a partial symmetric graph unchanged as, $A_{i,j} = A_{i,j}^T, D(G) = D(G')$. This leads to the following lemma.

\begin{lemma}\label{partial vs degree sym}
Every partial symmetric graph $G$ is degree symmetric.
\end{lemma}

The converse  of the  Lemma \ref{partial  vs degree  sym} need  not be
true. There   are  many  graphs  which   are  degree
  symmetric but  not partially  symmetric. For example,  consider the
graph depicted  in  Example  \ref{degree  symmetric
  graph}.

The  above  lemma  leads  us to  the  next  theorem.
It is a  sufficient condition  for separability  of
density matrices  arising from  partial symmetric graphs.   We mention
that, this  result generalizes  the result  of \cite{wu2006},  where a
similar result was obtained for a $2 \times n$ dimensional system.

\begin{theorem} \label{main}
Let $G$ be a partially symmetric graph with the following properties.
\begin{itemize}
\item
Between two vertices  of any partition $C_i$ there is  no edge. $(v_{i,l}, v_{i,k}) \notin E(G)$ for all $i, l, k$.
\item
Either there is no edge between vertices of $C_i$ and $C_j,$ or $A_{i,j}=A_{k,l}$ for all $i,j,k,l, i\neq j$ and $k\neq l.$
\item
Degrees of  all the vertices in a layer are same, i.e., $d_{C_i}(v_r) =
d_{C_i}(v_s)$ for all $v_r, v_s \in C_i,$ for all $i$.
\end{itemize}
Then  $\rho(G)$  is separable  i.e.  $\rho(G)  = \sum_i  w_i  \rho_A^i
\otimes \rho_B^i, \sum_i w_i = 1$.
\end{theorem}
Its proof is deferred to the appendix.

\begin{ex}\label{H}
An     example    of     a    partially     symmetric    graph     $H$
satisfying  all the conditions of
Theorem   \ref{main}  is   as  follows.   $$H  =   
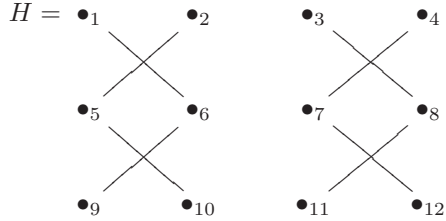
\xymatrix{\bullet_1
  \ar@{-}[dr]  &  \bullet_2  \ar@{-}[dl]  &  \bullet_3  \ar@{-}[dr]  &
  \bullet_4   \ar@{-}[dl]  \\   \bullet_5   \ar@{-}[dr]  &   \bullet_6
  \ar@{-}[dl]   &  \bullet_7   \ar@{-}[dr]  &   \bullet_8  \ar@{-}[dl]
  \\ \bullet_9 & \bullet_{10} & \bullet_{11} & \bullet_{12} }$$
\end{ex}

Theorem \ref{main} is a sufficient condition but not necessary. There are classes of partial symmetric graphs generating separable states without satisfying conditions of this theorem. Some of them will be discussed now.

Recall that, the union graph of two graphs $G = (V(G),E(G))$ and $H = (V(H),E(H))$ is defined as the new graph $G \cup H = (V(G) \cup V(H), E(G) \cup E(H))$ \cite{west2001}. Let $G$  be a graph of order $n$ with vertex labelling  $\{1, 2, \dots n\}$. Define, $mG = G  \cup G \cup \dots \cup G$ (union of $m$-copies of $G$) with vertex labelling $\{v_{j,k}: v_{j,k} = jn + k\}$. Note that, copies of $G$ form the layers of $mG$. There is no edge between two layers. Hence, $mG$ is trivially partially symmetric and it violates the $1$-st condition of Theorem \ref{main} which states that there will be no edge between two vertices located in the same layer. Interestingly, we will show now that $mG$ represents separable states. Observe that \begin{eqnarray*}
A(mG) &= \mbox{diag}\{A(G), A(G), \dots , A(G) (m ~\mbox{times})\} &= I_m \otimes A(G),\\
D(mG) &= \mbox{diag}\{D(G), D(G), \dots , D(G) (m ~\mbox{times})\} &= I_m \otimes D(G),\\
L(mG) &= \mbox{diag}\{L(G), L(G), \dots , L(G) (m ~\mbox{times})\} &= I_m \otimes L(G),\\
Q(mG) &= \mbox{diag}\{Q(G), Q(G), \dots , Q(G) (m ~\mbox{times})\} &= I_m \otimes Q(G).
\end{eqnarray*}
where, $I_m$ denotes the identity matrix of order $m$. Trivially, $\rho_l(mG) = \frac{L(mG)}{\mbox{trace}(L(mG))}$ and $\rho_q(mG) = \frac{Q(mG)}{\mbox{trace}(Q(mG))}$ are separable states. This result may be expressed as follows.
\begin{lemma}
For any graph  $G$, $\rho_l(mG)$ and $\rho_q(mG)$ represent separable bipartite states of dimension $m\times n$ w.r.t standard labelling on $mG$.
\end{lemma}

Note that, $G$ may not correspond to a separable state  but $mG$ always represents a separable state. This lemma is significant as it suggest more general conditions for separability.

We define a new graph operation as follows. Consider a partially symmetric graph $H$ with $m$ different layers, each layer having $n$ number of vertices and $H$ satisfies all the conditions of Theorem \ref{main}, while $G$ is a simple graph with $n$ vertices. We define the new graph $G\bowtie H$ as the graph which is constructed by replacing each layer of $H$ by the graph $G$. Note that, $V(G \bowtie H) = V(H)$. An example which illustrates the operation $G\bowtie H$, is given below.

\begin{ex}
Consider the star graph $G$ with four vertices given by
$$
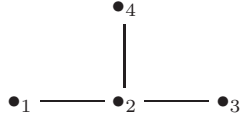
\xymatrix{& \bullet_4 \ar@{-}[d] & \\ \bullet_1 \ar@{-}[r] & \bullet_2 \ar@{-}[r] & \bullet_3} $$
and $H$ is a graph given in Example \ref{H}. Then the graph $G \bowtie H$ is as follows.
$$
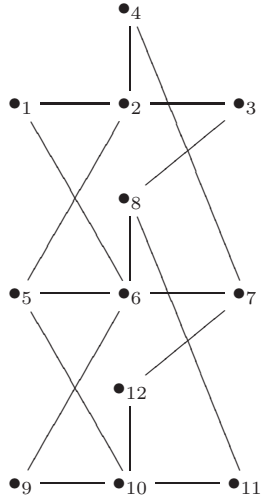
\xymatrix{& \bullet_4 \ar@{-}[d] \ar@{-}[rddd] & \\ \bullet_1 \ar@{-}[r] \ar@{-}[rdd] & \bullet_2 \ar@{-}[r] \ar@{-}[ldd]& \bullet_3 \ar@{-}[ld]\\ & \bullet_8 \ar@{-}[d] \ar@{-}[rddd] & \\ \bullet_5 \ar@{-}[r] \ar@{-}[rdd] & \bullet_6 \ar@{-}[r] \ar@{-}[ldd]& \bullet_7 \ar@{-}[ld] \\ & \bullet_{12} \ar@{-}[d] & \\ \bullet_9 \ar@{-}[r] & \bullet_{10} \ar@{-}[r] & \bullet_{11}}$$
\end{ex}

Now we present some properties of $G\bowtie H$ where $G$ and $H$ are the graphs as discussed above.

$H$ satisfies all the conditions of theorem (\ref{main}). Hence, there is no edge joining two vertices belonging to the same layer. This implies that the diagonal blocks of $A(H)$ are zero matrices. Graph $G$ is placed $m$ times on the layers of $H$. Thus all $m$ diagonal blocks of $A(G \bowtie H)$ are $A(G)$. Hence, we have the following lemma.

\begin{lemma}
$A(G \bowtie H) = A(mG) + A(H)$, where $I_m$ is the identity matrix of order $m$.
\end{lemma}

It is clear from the construction of $G \bowtie H$ that the degree of a vertex in $G \bowtie H$ is the sum of its degree in $H$ and its degree in $G$. Incorporating this in the expression of $A(G \bowtie H)$, we obtain the following result.
\begin{lemma}
$D(G \bowtie H) = D(mG) + D(H)$.
\end{lemma}

The above two lemmas together imply the structure of the Laplacian $L(G)$ and the signless Laplacian $Q(G)$, i.e., the structures of the density matrices $\rho_l(G)$ and $\rho_q(G)$.

\begin{lemma}
$L(G\bowtie H) = L(H) + L(mG)$ and
$Q(G\bowtie H) = Q(H) + Q(mG)$.
\end{lemma}
\begin{proof}

\begin{align*}
L(G\bowtie H) & = D(G\bowtie H) - A(G\bowtie H)\\
& = I_m \otimes D(G) + D(H) - I_m \otimes A(G) - A(H)\\
& = I_m \otimes (D(G) - A(G)) + D(H) - A(H)\\
& = I_m \otimes L(G) + L(H)\\
& = L(H) + L(mG)
\end{align*}
\end{proof}
Similarly, $Q(G\bowtie H) = Q(H) + Q(mG)$.

All the above lemmas together indicate the separability of $G\bowtie H$.
\begin{theorem}
$G\bowtie H$ represents a bipartite separable state of dimension $m \times n.$
\end{theorem}

\begin{ex}
Consider the Werner  state which is a  mixture  of  projectors  onto the  symmetric  and antisymmetric subspaces, with the  relative weight $p_{sym}$ being the only parameter that defines the state.
$$\rho(d, p_{sym}) =  p_{sym} \frac{2}{d^2 + d}  P_{sym} + (1-p_{sym}) \frac{2}{d^2 - d} P_{as},$$
where, $P_{sym} = \frac{1}{2}(1+P), P_{as} =   \frac{1}{2}(1-P)$,  are   the  projectors   and  $P   =  \sum_{ij} \ket{i}\bra{j}\otimes \ket{j}\bra{i}$ is the permutation operator that exchanges the two subsystems.

Only $\rho(d,0) =  \frac{I - P}{d^2 - d}$ is  represented by a Laplacian matrix of simple graphs. For example, for $d=2, 3$ we have

\begin{align*}
\rho(2,0) &= \begin{bmatrix} 0 & 0 & 0 & 0 \\ 0 & .5 & -.5 & 0 \\ 0 & -.5 & .5 & 0 \\ 0 & 0 & 0 & 0\end{bmatrix}\\
& \equiv 
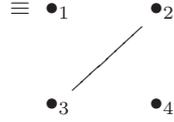
\xymatrix{\bullet_1 & \bullet_2 \ar@{-}[dl]\\ \bullet_3 & \bullet_4}\end{align*}
\begin{align*}
\rho(3,0) &= \begin{bmatrix} 0 & 0 & 0 & 0 & 0 & 0 & 0 & 0 & 0 \\
         0 &   0.1667 &        0 &  -0.1667 &        0 &        0 &        0 &        0 &        0\\
         0 &        0 &   0.1667 &        0 &        0 &        0 &  -0.1667 &        0 &        0\\
         0 &  -0.1667 &        0 &   0.1667 &        0 &        0 &        0 &        0 &        0\\
         0 &        0 &        0 &        0 &        0 &        0 &        0 &        0 &        0\\
         0 &        0 &        0 &        0 &        0 &   0.1667 &        0 &  -0.1667 &        0\\
         0 &        0 &  -0.1667 &        0 &        0 &        0 &   0.1667 &        0 &        0\\
         0 &        0 &        0 &        0 &        0 &  -0.1667 &        0 &   0.1667 &        0\\
         0 &        0 &        0 &        0 &        0 &        0 &        0 &        0 &        0\\
         \end{bmatrix}\\
& \equiv 
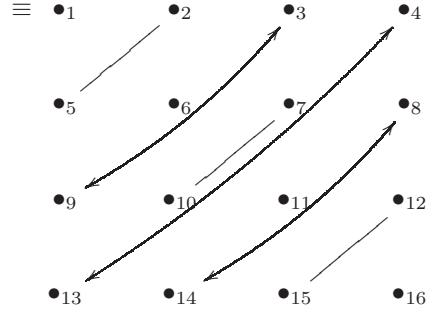
\xymatrix{\bullet_1 & \bullet_2 \ar@{-}[ld] & \bullet_3 \ar@/^/[ddll] & \bullet_4 \ar@/^/[dddlll]\\
					\bullet_5 & \bullet_6 & \bullet_7 \ar@{-}[ld] & \bullet_8 \ar@/^/[ddll] \\
					\bullet_9 \ar@/_/[uurr] & \bullet_{10} & \bullet_{11} & \bullet_{12} \ar@{-}[ld]\\
					\bullet_{13} \ar@/_/[uuurrr] & \bullet_{14}\ar@/_/[uurr] & \bullet_{15} & \bullet_{16} \\ } \\
\end{align*}

It is easy to verify that these graphs are not degree symmetric and hence not partially symmetric. Further, these graphs represent entangled states.
\end{ex}

\section{A non-local quantum operation on graphs}

Observe that the definition of partially symmetric graphs relies on the labelling of the vertices. In fact, in the graph theoretic approach of interpretation of quantum states, it is well known that properties of a density matrix derived from a graph is vertex labelling contingent. A graph which represents a separable state corresponding to a vertex labelling, may also produce an entangled state for a different vertex labelling. In this section we describe graph isomorphism as a non-local operation to generate entanglement. We begin with an example.
\begin{ex}\label{sepeng}
Let $G$ be a graph given by $$\xymatrix{\bullet \ar@{-}[r] & \bullet \ar@{-}[r] & \bullet \ar@{-}[r] & \bullet}$$ It is easy to verify that the density matrix $\rho_l(G_1)$ corresponding to the graph $G_1$ with labelled vertices given below represents a separable state.
$$
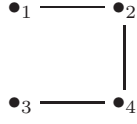
\xymatrix{\bullet_1 \ar@{-}[r] & \bullet_2 \ar@{-}[d]\\ \bullet_3 \ar@{-}[r] & \bullet_4}$$
Whereas, $\rho_l(G_2)$ represents an entangled state for the following graph $G_2$ with a different vertex labelling.
$$
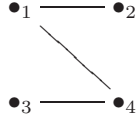
\xymatrix{\bullet_1 \ar@{-}[r] \ar@{-}[dr] & \bullet_2 \\ \bullet_3 \ar@{-}[r] & \bullet_4}$$\end{ex}

It is evident that these graphs are isomorphic. It has also been proved that separability of $\rho_l(G)$ when $G$ is a completely connected simple graph does  not depend on vertex  labelling, and the states $\rho_l(G)$ corresponding to a star graph with respect to any labelling are entangled [Section 6, \cite{braunstein2006}]. In \cite{braunstein2006}, it is also asked if there exist any other graphs which have this property. We mention that, in the search of partially symmetric graphs, we found one more graph given below, having the property that, for any vertex labelling, the graph represents an entangled state. In fact, this graph has no vertex labelling for which it can be made a partially symmetric graph.

\begin{ex} \label{asymmetric}
The graph $G$ for which no vertex labelling produces a partially symmetric graph.
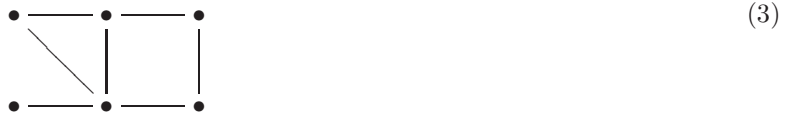
\begin{equation}
\xymatrix{\bullet \ar@{-}[r] \ar@{-}[dr] & \bullet \ar@{-}[r] \ar@{-}[d] & \bullet \ar@{-}[d] \\ \bullet & \bullet \ar@{-}[l] \ar@{-}[r] & \bullet}
\end{equation}
\end{ex}

Based on the above observations, we classify the set of all graphs with a fixed number of vertices into the following three classes.
\begin{enumerate}
\item
{\bf E-graph:} Independent of vertex labelling, all quantum states related to this graph are \textbf{E}ntangled.
\item
{\bf S-graph:} Independent of vertex labelling, all quantum states related to this graph are \textbf{S}eparable.
\item
{\bf ES-graph:} Quantum states related to some of the vertex labelling are \textbf{E}ntangled and others are \textbf{S}eparable.
\end{enumerate}
Obviously the completely connected graph is a S-graph, the star graph is an E-graph and the graph in Example \ref{asymmetric} is an E-graph.

In this section, we are interested in ES-graphs. These graphs provide a platform for generating entanglement using graph isomorphism as a non-local operation. Changing the vertex labelling on a graph representing a separable state, generates its isomorphic copy representing an entangled state. It is proved in the literature \cite{braunstein2006, adhikari2012}, that any graph with more than one edge represents a mixed state. Hence, graph isomorphism acts as an  entanglement generator on both pure and mixed states. For example, the isomorphism $\phi : V(G_1) \rightarrow V(G_2)$ defined as $$\phi(1) = 2, \phi(2) = 1, \phi(3) = 3, \phi(4) = 4$$ act as an mixed state entanglement generator in Example \ref{sepeng}. The following example of pure state entanglement generator may be of interest to the quantum information community.
\begin{ex}
The following graph represents the density matrix of the separable state $\frac{1}{\sqrt{2}}\ket{0 + 1}\ket{1}$
$$G_1 = 
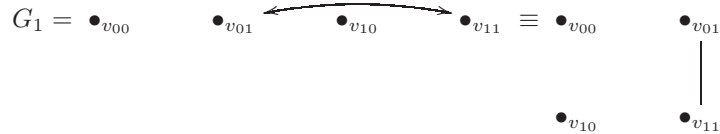
\xymatrix{\bullet_{v_{00}} & \bullet_{v_{01}}\ar@/^/[rr] & \bullet_{v_{10}}  & \bullet_{v_{11}}\ar@/_/[ll]} \equiv \xymatrix{\bullet_{v_{00}} & \bullet_{v_{01}} \ar@{-}[d] \\ \bullet_{v_{10}} & \bullet_{v_{11}}}$$
We define a graph isomorphism $\phi$ acting on $G_1$. $\phi(v_{00}) = v_{00}, \phi(v_{01}) = v_{01}, \phi(v_{10}) = v_{11}, \phi(v_{11}) = v_{10}$ . It generates the graph
$$G_2  = 
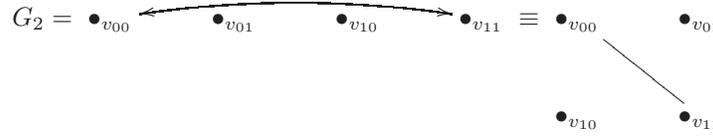
\xymatrix{\bullet_{v_{00}}  \ar@/^/[rrr]& \bullet_{v_{01}}  &
  \bullet_{v_{10}}   &   \bullet_{v_{11}}    \ar@/_/[lll]   }   \equiv
\xymatrix{\bullet_{v_{00}}     \ar@{-}[dr]      &     \bullet_{v_{01}}
  \\ \bullet_{v_{10}} & \bullet_{v_{11}}}$$ Graph $G_2$ represents the
Bell state $\frac{1}{\sqrt{2}}\ket{00 + 11}$ \cite{adhikari2012}. Note that graph $G_1$
was partially symmetric but graph  $G_2$ is not. The graph isomorphism
$\phi$  here acts in a
  fashion analogous to  a CNOT gate. Note that, every graph isomorphism corresponds to permutation similar matrices (for example, Laplacian and signless Laplacian matrices) associated with the graph and its isomorphic copy. This has a resemblance to a CNOT gate which is  itself a permutation
matrix. At  the end of  this section we  present an example  where the
permutation   matrix   is   different   from   the   CNOT   operation.
Thus, we  may conclude that graph  isomorphisms are in
  general entangling operations.
\end{ex}

These  examples  inspire  a number  of  questions  for
  further investigation.  For instance, which isomorphisms  will act as
  an entanglement generator? In the remaining part of this work we try
  to address this question.
\begin{de}
In a graph $G$, partial degree of  a vertex $v_{i,k} \in C_i$, w.r.t the layer $C_j$  is denoted  by $ld_{C_j}(v_{i,k})_G$  and defined  by the number of  edges from $v_{i,k}$  to the  vertices of $C_j$.   When no confusion  occurs, instead  of $ld_{C_j}(v_{i,k})_G$,  we may  write $ld_{C_j}(v_{i,k})$.
\end{de}

\begin{de}
In a  graph $G$, a  vertex $v_{i,k}$  is internally related  to vertex $v_{i,l}$  in  $C_i$  w.r.t  layer $C_j$  if  $(v_{i,k},  v_{j,l})$  and $(v_{i,l},v_{j,k}) \in E(G)$.
\end{de}
$ld_{C_j}(v_{i,k})  =  $  number  of vertices  internally  related  to $v_{i,k}$ in $C_i$ w.r.t $C_j$ for a partial symmetric graph.

\begin{de}
$G$ is  called partially asymmetric  if $\exists ~  (v_{i,k}, v_{j,l}) \in E(G), i \neq j, k \neq l$ such that $(v_{i,l}, v_{j,k}) \notin E(G)$.
\end{de}

Incidence set of $v \in V(G)$ is $I_G(v) = \{w : (w,v) \in E(G)\}$, that is, set of all vertices incident to vertex $v$. Incidence interchange between two vertices $u, v$, denoted by $u \overset{i}{\leftrightarrow} v$, is a graphical operation to construct a graph $H$ from $G$, defined as follows.
\begin{equation}
u \overset{i}{\leftrightarrow} v \equiv \begin{cases} I_H(u) = I_G(v) \\ I_H(v) = I_G(u) \end{cases}.
\end{equation}
This  operation can generate  mixed  entangled  states from  a  mixed separable state, as described later.  Note that this is not a physical operation   between  two   pre-existing   particles,   but  a   purely mathematical operation between two ``formal'' particles induced by how we biparition  the graph.   Hence there is  no contradiction  with the physical principle  of non-increase of entanglement  under LOCC (local operations and  classical communication).   This form  of entanglement creation is reminiscent of the  idea put forth in \cite{zanardi}, that the degrees of freedom and hence entanglement are observer-induced.

Note  that $H$  is a  layered graph  and is also isomorphic  to $G$. Let us see an example.
\begin{ex}\label{exp:qo}
Initially we consider a graph $G$ with the following labels and layers $C_1=\{1, 2, 3\}$ and $C_2=\{4, 5, 6\}$
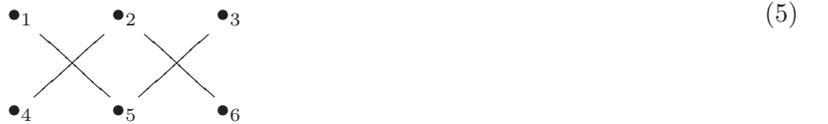
\begin{equation}
\xymatrix{\bullet_1 \ar@{-}[dr] & \bullet_2 \ar@{-}[dr] \ar@{-}[dl] & \bullet_3 \ar@{-}[dl] \\ \bullet_4 & \bullet_5 & \bullet_6}
\end{equation}
$H$ is  generated from $G$  by graphical operation  $1 \leftrightarrow 2$, namely:
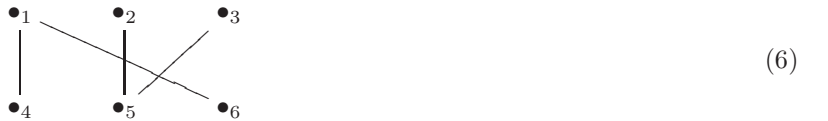
\begin{equation}
\begin{array}{cc}
\xymatrix{\bullet_1 \ar@{-}[d] \ar@{-}[drr] & \bullet_2 \ar@{-}[d] & \bullet_3 \ar@{-}[dl] \\ \bullet_4 & \bullet_5 & \bullet_6}
\end{array}
\end{equation}
\end{ex}

Note that in the above example initially $G$ was a partially symmetric graph. $dl_{C_1}(1)  = 1$ but  $dl_{C_1}(2) = 2$.  After interchanging the  vertex labellings  of $1$  and  $2$ the new  graph is  $H$, which  is partially asymmetric. It can be generalised for an arbitrary partially symmetric graph.

Let in a partially symmetric graph $G$, $ld_{C_j}(v_{i,k}) \ge ld_{C_j}(v_{i,l})$, then $ld_{C_j}(v_{i,k}) - ld_{C_j}(v_{i,l})      \ge     1$.       $ld_{C_j}(v_{i,k})$ and $ld_{C_j}(v_{i,l})$ represent number of internally related vertices in $C_i$  of $v_{i,k}$  and $v_{i,l}$ w.r.t $C_j$, respectively.  Thus, there exists  at  least  one  vertex  $v_{i,s}$,  internally  related  to $v_{i,k}$  but   not  with  $v_{i,k}$.   After   interchanging  vertex labellings  there will  be at  least  one edge  incident to  $v_{i,s}$ without  any  complement  as  the  complement  edge  is  misplaced  by interchange. Hence, the new graph $H$, isomorphic to $G$, is partially asymmetric. This can be expressed as a lemma.

\begin{lemma}
Assume, $ld_{C_j}(v_{i,k}) \neq  ld_{C_j}(v_{i,l})$ in a partially symmetric graph $G$. Graph $H$ is generated after interchanging vertex labellings of the vertices $v_{i,l}$ and $v_{i,k} \in E(G)$. Then $H$ is partially asymmetric.
\end{lemma}

Also, $(v_{i,s}, v_{j,l})\notin  E(G) \Rightarrow$ complement  of $(v_{i,l},   v_{j,k}) \notin E(H)$, but  $(v_{i,l}, v_{j,k}) \in E(H)$. Trivially,   $H$ is not partially symmetric.
\begin{lemma}
Let $(v_{i,l}, v_{j,k}) \in E(G), i \neq j, l \neq k$; but $(v_{i,s}, v_{j,l})\notin E(G)$ for some $s$, where $G$ is a partially symmetric graph. Interchange of vertex labellings of $v_{i,s}$ and $v_{i,l}$ will generate partial asymmetric graph $H$.
\end{lemma}

This change of labelling may not generate partial asymmetry in all the cases. Suppose  any two  vertices of  $C_i$  are not  internally related  w.r.t $C_j$. Hence, any edge between  vertices of $C_i$ and  $C_j$ is of the  form $(v_{i,k}, v_{j,k}) ~\forall~  k = 1, 2,  \dots n$. Consider any  two vertices of $C_i$, say  $v_{i,l}$ and $v_{i,k}$. Interchange  of vertex labellings of these two vertices will generate new edges $(v_{i,l}, v_{j,k})$ and $(v_{i,k}, v_{j,l})$. This implies partial symmetry in the new graph. We may write it as a lemma.
\begin{lemma}
Suppose  any two  vertices of  $C_i$  are not  internally related  w.r.t $C_j$.   Also  assume that  $ld_{C_j}(v_{i,l})  =  ld_{C_j}(v_{i,k}) ~\forall~ k,l = 1, 2, \dots n$.  Then interchange of vertex labellings of any two vertices of $C_i$ will not generate partial asymmetry.
\end{lemma}

Graph  isomorphism is an  equivalence relation  on the set  of all simple  graphs,   which  forms  disjoint  equivalence   classes.   Let $\mathcal{G}$  be one  such  class  and $\mathcal{L}$  be  set of  all isomorphisms on $\mathcal{G}$.  $\circ$ is composition of mappings.

Trivially $(\mathcal{L},\circ)$ forms  a group which is a permutation  group  over  $\#(V(G))$   elements.  For  an ES  graph $\mathcal{G}   =  \mathcal{E}   \cup  \mathcal{S},   \mathcal{E}  \cap \mathcal{S}   =  \phi,   \mathcal{E}  \neq   \phi,  \mathcal{S}   \neq \phi$. $\mathcal{E}$ and $\mathcal{S}$ are subclasses of $\mathcal{G}$ consisting  of all  graphs  providing entangled  and separable  states respectively.

Let  $\mathcal{L}_e$   and  $\mathcal{L}_s$   be  the  group  of all  graph isomorphisms  acting  on  $\mathcal{E}$ and  $\mathcal{S}$.  Trivially $(\mathcal{L}_e,  \circ)$  and   $(\mathcal{L}_s,  \circ)$  also  form groups. Entanglement generators are invertible mappings from $(\mathcal{L}_s, \circ)$ to $(\mathcal{L}_e, \circ)$.

\begin{remark}
In example \ref{exp:qo}, the graphical operation $1\leftrightarrow 2$ represents a quantum entanglement generator which transforms the separable states $\rho(G)$ to entangled states $\rho(H).$
\end{remark}

\begin{ex}
It is clear to us that graph isomorphism acts as a global unitary operator and it is capable to generate mixed entangled state from mixed separable state. Consider two isomorphic graphs.
$$
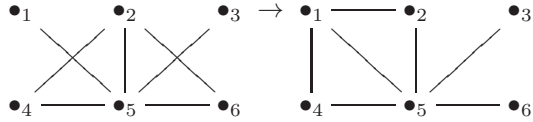
\xymatrix{\bullet_1\ar@{-}[dr] & \bullet_2 \ar@{-}[d] \ar@{-}[dr] \ar@{-}[dl] & \bullet_3 \ar@{-}[dl]\\ \bullet_4 \ar@{-}[r] & \bullet_5 \ar@{-}[r] & \bullet_6} \rightarrow \xymatrix{\bullet_1 \ar@{-}[dr] \ar@{-}[d] \ar@{-}[r] & \bullet_2 \ar@{-}[d]& \bullet_3 \ar@{-}[dl] \\ \bullet_4 \ar@{-}[r] & \bullet_5 \ar@{-}[r] & \bullet_6}$$
Corresponding permutation is
$$\begin{pmatrix} 1 & 2 & 3 & 4 & 5 & 6 \\ 6 & 1 & 3 & 4 & 5 & 2\end{pmatrix}$$
The permutation matrix is
$$\begin{bmatrix} 0 & 1 & 0 & 0 & 0 & 0 \\ 0 & 0 & 0 & 0 & 0 & 1 \\ 0 & 0 & 1 & 0 & 0 & 0 \\ 0 & 0 & 0 & 1 & 0 & 0 \\ 0 & 0 & 0 & 0 & 1 & 0 \\ 1 & 0 & 0 & 0 & 0 & 0 \end{bmatrix}.$$
This operator acts as an entanglement generator. Density matrices corresponding to the first graph is separable but for the second graph $\rho_l$ and $\rho_q$ both are entangled.
\end{ex}

\section{Conclusion and open problems}
The quantum  separability problem is  an
    important  and  difficult  open  problem  in  quantum  information
    theory.
  For  quantum  states related  to  simple  combinatorial graphs  some
  sufficiency  conditions   are  available  in  the   literature.  For
  bipartite systems they were applicable  for some special cases of $2
  \times p$  systems.  Here, we  have generalised these results  to $m
  \times n$ systems.
  
In another direction, ourwork proposes the use of
  of graph isomorphisms as entanglement generators, which can generate
  mixed  entangled states  from mixed  separable states. Note that, these isomorphisms are formal operations, in contrast to physical operations like LOCC (local operations and classical communication), which cannot generate entanglement.
As mentioned
  above,  combinatorial  graphs  enable  us to  visualize  changes  of
  quantum states under a  particular quantum operation  pictorially.  In
  this context,  graph isomorphisms  pictorially depict certain actions that lead to entanglement generation.
Finally, this  work initiates a  number of
  problems or directions for future investigations:
(a). Can  a combinatorial  criterion be  defined to  detect entangled
  states  arising from  graphs?  Can  the quality  of entanglement  be
  defined by using the partially asymmetric graphs?
(b). Can the formulation of partially symmetric graphs be generalized
  for  weighted  graphs  that   may  possibly  open  up  combinatorial
  formulation of separable states?
(c). Generalization of the bipartite separability criteria  arising from partially
  symmetric  graphs to the case of multipartite
  states?
(d). Further   investigations   are   required   for   the
  identification of  ES graphs (See  example 7). Precisely, when  is a
  graph an  ES graph? How  much entanglement  can be generated  from a
  separable   copy  of   an   ES  graph   using  graph   isomorphism?
Here the results of \cite{SPS+05} should be leveraged.

We hope that this work is  a contribution to the graphical
  representation of quantum mechanics, in general and the separability
  problem, in particular.

\section*{Acknowledgement}
This  work is  partially  supported by  the  project ``Graph  theoretical aspects of quantum information processing" (Grant No. 25(0210)/13/EMR II)  funded by  ``Council of Scientific and Industrial Research, New Delhi". SD is  thankful  to the Ministry of Human Resource Development, Government of the Republic of India, for his doctoral fellowship. We would like to thank the anonymous referee for his constructive comments.

\section*{Appendix}

\textbf{Proof of Theorem \ref{theorem1}:}

\begin{proof}
Let $G$ be a graph and $\rho_l(G)$ separable. Then
\begin{align*}
\rho_l(G) & = \sum_i p_i \rho_i^A \otimes \rho_i^B\\
\Rightarrow \rho_l(G)^{T_B} &= \sum_i p_i \rho_i^A \otimes (\rho_i^B)^{T_B}\\
&= \frac{1}{\tr(L(G))}(L(G))^{T_B} = \frac{1}{\tr(L(G))}((D(G))^{T_B} - (A(G))^{T_B})\\
&= \frac{1}{\tr(L(G))}(D(G) - A(G'))\\
&= \frac{1}{\tr(L(G))}(D(G) - D(G') + D(G') - A(G'))\\
&= \frac{1}{\tr(L(G))}(D(G) - D(G') + L(G'))\\
&= \frac{1}{\tr(L(G'))}L(G') + \frac{1}{\tr(L(G))}(D(G) - D(G')) ~[\because d(G) = d(G') \Rightarrow \tr(L(G)) = \tr(L(G')).]\\
\end{align*}
\begin{align*}
\rho_l(G') &= \rho_l(G)^{T_B} - \frac{1}{\tr(L(G))}(D(G) - D(G'))\\
&= \sum_i p_i \rho_i^A \otimes (\rho_i^B)^{T_B} - \frac{1}{\tr(L(G))}(D(G) - D(G'))\\
\rho_l(G')^{T_B} &= \sum_i p_i \rho_i^A \otimes \rho_i^B - \frac{1}{\tr(L(G))}(D(G) - D(G')) ~[\because (D(G))^{T_B} = D(G).]
\end{align*}
Thus, the desired result follows for $\rho_l(G).$\\
Similarly, $\rho_q(G')^{T_B} = \sum_i p_i \rho_i^A \otimes \rho_i^B + \frac{1}{\tr(Q(G))}(D(G) - D(G'))$, assuming, $\rho_q(G') = \sum_i p_i \rho_i^A \otimes \rho_i^B$. This completes the proof.
\end{proof}

\textbf{Proof of theorem \ref{main}:}

\begin{proof}
Since $A_{i,j}$ is a symmetric matrix, the spectral decomposition of $A_{i,j}$ is given by $A_{i,j} = \sum_r \lambda_r u_r u_r^t$ where $\{u_r : r=1:n\}$ is a
complete set of orthonormal eigenvectors corresponding to the eigenvalues $\lambda_r, r=1:n$ of $A_{i,j}.$  For $A_{i,j} = 0, A_{i,j} = \sum_r 0. u_r u_r^t$.
Since $u_r, r=1:n$ are normalised eigenvectors, $u_ru_r^t$ is a trace $1$ positive semi-definite matrix for each $r$. Since there are no edges between any two vertices
n any layer $C_i,$ $A_{i} = 0$ for all $i$. Further, $D_i = diag\{d_i\} = d_i I$ since $A_{i,j}=A_{k,l}$ for all $i,j,k,l, i\neq j$ and $k\neq l.$\\
Then
\begin{align*}
L(G) & = \begin{bmatrix}
		d_0.I & A_{0,1} & A_{0,2} & \dots & A_{0,(m-1)} \\
		A_{0,1} & d_1.I & A_{0,2} & \dots & A_{0,(m-1)} \\
		\vdots & \vdots & \vdots & \vdots & \vdots\\
		A_{0,(m-1)} & A_{1,(m-1)} & A_{2,(m-1)} & \dots & d_{m-1}.I
		\end{bmatrix}\\
		& = \begin{bmatrix}
		d_0\sum_r u_r u_r^t & \sum_r\lambda_r u_r u_r^t & \sum_r\lambda_r u_r u_r^t & \dots & \sum_r\lambda_r u_r u_r^t\\
		\sum_r\lambda_r u_r u_r^t & d_1\sum_r u_r u_r^t & \sum_r\lambda_r u_r u_r^t & \dots & \sum_r\lambda_r u_r u_r^t\\
		\vdots & \vdots & \vdots & \vdots & \vdots\\
		\sum_r\lambda_r u_r u_r^t & \sum_r\lambda_r u_r u_r^t & \sum_r\lambda_r u_r u_r^t & \dots & d_{(m-1)}\sum_r u_r u_r^t
		\end{bmatrix}\\
		& = \sum_r\begin{bmatrix}
		d_0 & \lambda_r & \lambda_r & \dots & \lambda_r\\
		\lambda_r& d_1 & \lambda_r & \dots & \lambda_r\\
		\vdots & \vdots & \vdots & \vdots & \vdots\\
		\lambda_r & \lambda_r & \lambda_r & \dots & d_{m-1}\\
		\end{bmatrix} \otimes u_ru_r^t\\
		& = \sum_r B(r) \otimes u_r u_r^t,
\end{align*}
where $B(r) = \begin{bmatrix}
		d_0 & \lambda_r & \lambda_r & \dots & \lambda_r\\
		\lambda_r& d_1 & \lambda_r & \dots & \lambda_r\\
		\vdots & \vdots & \vdots & \vdots & \vdots\\
		\lambda_r & \lambda_r & \lambda_r & \dots & d_m\\
		\end{bmatrix}$. Note that $A_{i,j} = 0 \Rightarrow b_{i,j} = 0$. Now we want to show $B$ is a positive semidefinite matrix. \\
Note that the spectral radius of $A_{i,j} \le \|A_{i,j}\|_{\infty}$, where $\|A_{i,j}\|_{\infty}$ is the subordinate matrix norm defined by $\|A_{i,j}\|_{\infty}= \max_{i}\sum_{j = 1}^n|a_{i,j}|$. Besides, $d_i = \sum_{k = 0}^{m-1}\max_{i}\sum_{j = 1}^n|a_{i,j}| = m\max_{i}\sum_{j = 1}^n|a_{i,j}|$. Then, $(m-1)\lambda_r \le (m-1) \times($ Spectral radius of $A_{i,j}) \le d_i ~\forall~ i$. Hence, $B$ is a diagonally dominant symmetric matrix with all positive entries. So $B$ is a positive semidefinite matrix. Hence $\rho_l(G)$ is separable. Similarly the result follows for $\rho_q(G).$
\end{proof}

% In time of submission comment library and uncomment partial.bbl. partial.bbl shall be included in the submission to generate the reference part of submission.
%\bibliography{library}

\begin{thebibliography}{47}
\expandafter\ifx\csname natexlab\endcsname\relax\def\natexlab#1{#1}\fi
\expandafter\ifx\csname bibnamefont\endcsname\relax
  \def\bibnamefont#1{#1}\fi
\expandafter\ifx\csname bibfnamefont\endcsname\relax
  \def\bibfnamefont#1{#1}\fi
\expandafter\ifx\csname citenamefont\endcsname\relax
  \def\citenamefont#1{#1}\fi
\expandafter\ifx\csname url\endcsname\relax
  \def\url#1{\texttt{#1}}\fi
\expandafter\ifx\csname urlprefix\endcsname\relax\def\urlprefix{URL }\fi
\providecommand{\bibinfo}[2]{#2}
\providecommand{\eprint}[2][]{\url{#2}}

\bibitem[{\citenamefont{West et~al.}(2001)}]{west2001}
\bibinfo{author}{\bibfnamefont{D.~B.} \bibnamefont{West}} \bibnamefont{et~al.},
  \emph{\bibinfo{title}{Introduction to graph theory}},
  vol.~\bibinfo{volume}{2} (\bibinfo{publisher}{Prentice hall Upper Saddle
  River}, \bibinfo{year}{2001}).

\bibitem[{\citenamefont{Bapat}(2010)}]{bapat2010}
\bibinfo{author}{\bibfnamefont{R.~B.} \bibnamefont{Bapat}},
  \emph{\bibinfo{title}{Graphs and matrices}} (\bibinfo{publisher}{Springer},
  \bibinfo{year}{2010}).

\bibitem[{\citenamefont{Northrop}(2010)}]{northrop2010introduction}
\bibinfo{author}{\bibfnamefont{R.~B.} \bibnamefont{Northrop}},
  \emph{\bibinfo{title}{Introduction to complexity and complex systems}}
  (\bibinfo{publisher}{CRC Press}, \bibinfo{year}{2010}).

\bibitem[{\citenamefont{Chung and Lu}(2006)}]{chung2006complex}
\bibinfo{author}{\bibfnamefont{F.~R.} \bibnamefont{Chung}} \bibnamefont{and}
  \bibinfo{author}{\bibfnamefont{L.}~\bibnamefont{Lu}},
  \emph{\bibinfo{title}{Complex graphs and networks}}, vol.
  \bibinfo{volume}{107} (\bibinfo{publisher}{American mathematical society
  Providence}, \bibinfo{year}{2006}).

\bibitem[{\citenamefont{L{\"u} and Zhou}(2011)}]{lu2011link}
\bibinfo{author}{\bibfnamefont{L.}~\bibnamefont{L{\"u}}} \bibnamefont{and}
  \bibinfo{author}{\bibfnamefont{T.}~\bibnamefont{Zhou}},
  \bibinfo{journal}{Physica A: Statistical Mechanics and its Applications}
  \textbf{\bibinfo{volume}{390}}, \bibinfo{pages}{1150} (\bibinfo{year}{2011}).

\bibitem[{\citenamefont{Han}(2012)}]{han2012graph}
\bibinfo{author}{\bibfnamefont{L.}~\bibnamefont{Han}}, Ph.D. thesis,
  \bibinfo{school}{University of York,} (\bibinfo{year}{2012}).

\bibitem[{\citenamefont{Bianconi}(2015)}]{bianconi2015interdisciplinary}
\bibinfo{author}{\bibfnamefont{G.}~\bibnamefont{Bianconi}},
  \bibinfo{journal}{EPL (Europhysics Letters)} \textbf{\bibinfo{volume}{111}},
  \bibinfo{pages}{56001} (\bibinfo{year}{2015}).

\bibitem[{\citenamefont{Nielsen and Chuang}(2010)}]{nielsen2010quantum}
\bibinfo{author}{\bibfnamefont{M.~A.} \bibnamefont{Nielsen}} \bibnamefont{and}
  \bibinfo{author}{\bibfnamefont{I.~L.} \bibnamefont{Chuang}},
  \emph{\bibinfo{title}{Quantum computation and quantum information}}
  (\bibinfo{publisher}{Cambridge university press}, \bibinfo{year}{2010}).

\bibitem[{\citenamefont{Berkolaiko and
  Kuchment}(2013)}]{berkolaiko2013introduction}
\bibinfo{author}{\bibfnamefont{G.}~\bibnamefont{Berkolaiko}} \bibnamefont{and}
  \bibinfo{author}{\bibfnamefont{P.}~\bibnamefont{Kuchment}},
  \emph{\bibinfo{title}{Introduction to quantum graphs}}, \bibinfo{number}{186}
  (\bibinfo{publisher}{American Mathematical Soc.}, \bibinfo{year}{2013}).

\bibitem[{\citenamefont{Berkolaiko}(2006)}]{berkolaiko2006quantum}
\bibinfo{author}{\bibfnamefont{G.}~\bibnamefont{Berkolaiko}},
  \emph{\bibinfo{title}{Quantum Graphs and Their Applications: Proceedings of
  an AMS-IMS-SIAM Joint Summer Research Conference on Quantum Graphs and Their
  Applications, June 19-23, 2005, Snowbird, Utah}}, vol. \bibinfo{volume}{415}
  (\bibinfo{publisher}{American Mathematical Soc.}, \bibinfo{year}{2006}).

\bibitem[{\citenamefont{Hein et~al.}(2004)\citenamefont{Hein, Eisert, and
  Briegel}}]{hein2004multiparty}
\bibinfo{author}{\bibfnamefont{M.}~\bibnamefont{Hein}},
  \bibinfo{author}{\bibfnamefont{J.}~\bibnamefont{Eisert}}, \bibnamefont{and}
  \bibinfo{author}{\bibfnamefont{H.~J.} \bibnamefont{Briegel}},
  \bibinfo{journal}{Physical Review A} \textbf{\bibinfo{volume}{69}},
  \bibinfo{pages}{062311} (\bibinfo{year}{2004}).

\bibitem[{\citenamefont{Anders and Briegel}(2006)}]{anders2006fast}
\bibinfo{author}{\bibfnamefont{S.}~\bibnamefont{Anders}} \bibnamefont{and}
  \bibinfo{author}{\bibfnamefont{H.~J.} \bibnamefont{Briegel}},
  \bibinfo{journal}{Physical Review A} \textbf{\bibinfo{volume}{73}},
  \bibinfo{pages}{022334} (\bibinfo{year}{2006}).

\bibitem[{\citenamefont{Benjamin et~al.}(2006)\citenamefont{Benjamin, Browne,
  Fitzsimons, and Morton}}]{benjamin2006brokered}
\bibinfo{author}{\bibfnamefont{S.~C.} \bibnamefont{Benjamin}},
  \bibinfo{author}{\bibfnamefont{D.~E.} \bibnamefont{Browne}},
  \bibinfo{author}{\bibfnamefont{J.}~\bibnamefont{Fitzsimons}},
  \bibnamefont{and} \bibinfo{author}{\bibfnamefont{J.~J.}
  \bibnamefont{Morton}}, \bibinfo{journal}{New Journal of Physics}
  \textbf{\bibinfo{volume}{8}}, \bibinfo{pages}{141} (\bibinfo{year}{2006}).

\bibitem[{\citenamefont{Singh et~al.}(2005)\citenamefont{Singh, Pal, Kumar, and
  Srikanth}}]{SPS+05}
\bibinfo{author}{\bibfnamefont{S.~K.} \bibnamefont{Singh}},
  \bibinfo{author}{\bibfnamefont{S.~P.} \bibnamefont{Pal}},
  \bibinfo{author}{\bibfnamefont{S.}~\bibnamefont{Kumar}}, \bibnamefont{and}
  \bibinfo{author}{\bibfnamefont{R.}~\bibnamefont{Srikanth}},
  \bibinfo{journal}{Jl. Math. Phys.} \textbf{\bibinfo{volume}{46}},
  \bibinfo{pages}{122105} (\bibinfo{year}{2005}).

\bibitem[{\citenamefont{Pal et~al.}(2006)\citenamefont{Pal, Kumar, and
  Srikanth}}]{PSS06}
\bibinfo{author}{\bibfnamefont{S.~P.} \bibnamefont{Pal}},
  \bibinfo{author}{\bibfnamefont{S.}~\bibnamefont{Kumar}}, \bibnamefont{and}
  \bibinfo{author}{\bibfnamefont{R.}~\bibnamefont{Srikanth}}, in
  \emph{\bibinfo{booktitle}{Quantum Compution: Back Action}}, edited by
  \bibinfo{editor}{\bibfnamefont{D.}~\bibnamefont{Goswami}}
  (\bibinfo{year}{2006}), vol. \bibinfo{volume}{864}, pp.
  \bibinfo{pages}{156--170}.

\bibitem[{\citenamefont{Braunstein
  et~al.}(2006{\natexlab{a}})\citenamefont{Braunstein, Ghosh, and
  Severini}}]{braunstein2006}
\bibinfo{author}{\bibfnamefont{S.~L.} \bibnamefont{Braunstein}},
  \bibinfo{author}{\bibfnamefont{S.}~\bibnamefont{Ghosh}}, \bibnamefont{and}
  \bibinfo{author}{\bibfnamefont{S.}~\bibnamefont{Severini}},
  \bibinfo{journal}{Annals of Combinatorics} \textbf{\bibinfo{volume}{10}},
  \bibinfo{pages}{291} (\bibinfo{year}{2006}{\natexlab{a}}).

\bibitem[{\citenamefont{Adhikari et~al.}(2016)\citenamefont{Adhikari, Adhikari,
  Banerjee, and Kumar}}]{adhikari2012}
\bibinfo{author}{\bibfnamefont{B.}~\bibnamefont{Adhikari}},
  \bibinfo{author}{\bibfnamefont{S.}~\bibnamefont{Adhikari}},
  \bibinfo{author}{\bibfnamefont{S.}~\bibnamefont{Banerjee}}, \bibnamefont{and}
  \bibinfo{author}{\bibfnamefont{A.}~\bibnamefont{Kumar}},
  \bibinfo{journal}{sumbitted}  (\bibinfo{year}{2016}).

\bibitem[{\citenamefont{Ionicioiu and Spiller}(2012)}]{ionicioiu2012encoding}
\bibinfo{author}{\bibfnamefont{R.}~\bibnamefont{Ionicioiu}} \bibnamefont{and}
  \bibinfo{author}{\bibfnamefont{T.~P.} \bibnamefont{Spiller}},
  \bibinfo{journal}{Physical Review A} \textbf{\bibinfo{volume}{85}},
  \bibinfo{pages}{062313} (\bibinfo{year}{2012}).

\bibitem[{\citenamefont{Dutta et~al.}(2016)\citenamefont{Dutta, Adhikari, and
  Banerjee}}]{dutta2016}
\bibinfo{author}{\bibfnamefont{S.}~\bibnamefont{Dutta}},
  \bibinfo{author}{\bibfnamefont{B.}~\bibnamefont{Adhikari}}, \bibnamefont{and}
  \bibinfo{author}{\bibfnamefont{S.}~\bibnamefont{Banerjee}},
  \bibinfo{journal}{Quantum Information Processing}  (\bibinfo{year}{2016}),
  \bibinfo{note}{doi:10.1007/s11128-016-1250-y}.

\bibitem[{\citenamefont{Du et~al.}(2010)\citenamefont{Du, Li, Li, and
  Severini}}]{du2010note}
\bibinfo{author}{\bibfnamefont{W.}~\bibnamefont{Du}},
  \bibinfo{author}{\bibfnamefont{X.}~\bibnamefont{Li}},
  \bibinfo{author}{\bibfnamefont{Y.}~\bibnamefont{Li}}, \bibnamefont{and}
  \bibinfo{author}{\bibfnamefont{S.}~\bibnamefont{Severini}},
  \bibinfo{journal}{Linear Algebra and its Applications}
  \textbf{\bibinfo{volume}{433}}, \bibinfo{pages}{1722} (\bibinfo{year}{2010}).

\bibitem[{\citenamefont{Passerini and Severini}(2008)}]{passerini2008neumann}
\bibinfo{author}{\bibfnamefont{F.}~\bibnamefont{Passerini}} \bibnamefont{and}
  \bibinfo{author}{\bibfnamefont{S.}~\bibnamefont{Severini}},
  \bibinfo{journal}{Available at SSRN 1382662}  (\bibinfo{year}{2008}).

\bibitem[{\citenamefont{Zhao et~al.}(2011)\citenamefont{Zhao, Halu, Severini,
  and Bianconi}}]{zhao2011entropy}
\bibinfo{author}{\bibfnamefont{K.}~\bibnamefont{Zhao}},
  \bibinfo{author}{\bibfnamefont{A.}~\bibnamefont{Halu}},
  \bibinfo{author}{\bibfnamefont{S.}~\bibnamefont{Severini}}, \bibnamefont{and}
  \bibinfo{author}{\bibfnamefont{G.}~\bibnamefont{Bianconi}},
  \bibinfo{journal}{Physical Review E} \textbf{\bibinfo{volume}{84}},
  \bibinfo{pages}{066113} (\bibinfo{year}{2011}).

\bibitem[{\citenamefont{Anand and Bianconi}(2009)}]{anand2009entropy}
\bibinfo{author}{\bibfnamefont{K.}~\bibnamefont{Anand}} \bibnamefont{and}
  \bibinfo{author}{\bibfnamefont{G.}~\bibnamefont{Bianconi}},
  \bibinfo{journal}{Physical Review E} \textbf{\bibinfo{volume}{80}},
  \bibinfo{pages}{045102} (\bibinfo{year}{2009}).

\bibitem[{\citenamefont{Anand et~al.}(2011)\citenamefont{Anand, Bianconi, and
  Severini}}]{anand2011shannon}
\bibinfo{author}{\bibfnamefont{K.}~\bibnamefont{Anand}},
  \bibinfo{author}{\bibfnamefont{G.}~\bibnamefont{Bianconi}}, \bibnamefont{and}
  \bibinfo{author}{\bibfnamefont{S.}~\bibnamefont{Severini}},
  \bibinfo{journal}{Physical Review E} \textbf{\bibinfo{volume}{83}},
  \bibinfo{pages}{036109} (\bibinfo{year}{2011}).

\bibitem[{\citenamefont{Maleti{\'c} and
  Rajkovi{\'c}}(2012)}]{maletic2012combinatorial}
\bibinfo{author}{\bibfnamefont{S.}~\bibnamefont{Maleti{\'c}}} \bibnamefont{and}
  \bibinfo{author}{\bibfnamefont{M.}~\bibnamefont{Rajkovi{\'c}}},
  \bibinfo{journal}{The European Physical Journal Special Topics}
  \textbf{\bibinfo{volume}{212}}, \bibinfo{pages}{77} (\bibinfo{year}{2012}).

\bibitem[{\citenamefont{Rovelli and Vidotto}(2010)}]{rovelli2010single}
\bibinfo{author}{\bibfnamefont{C.}~\bibnamefont{Rovelli}} \bibnamefont{and}
  \bibinfo{author}{\bibfnamefont{F.}~\bibnamefont{Vidotto}},
  \bibinfo{journal}{Physical Review D} \textbf{\bibinfo{volume}{81}},
  \bibinfo{pages}{044038} (\bibinfo{year}{2010}).

\bibitem[{\citenamefont{Cvetkovi{\'c} et~al.}(2007)\citenamefont{Cvetkovi{\'c},
  Rowlinson, and Simi{\'c}}}]{cvetkovic}
\bibinfo{author}{\bibfnamefont{D.}~\bibnamefont{Cvetkovi{\'c}}},
  \bibinfo{author}{\bibfnamefont{P.}~\bibnamefont{Rowlinson}},
  \bibnamefont{and} \bibinfo{author}{\bibfnamefont{S.~K.}
  \bibnamefont{Simi{\'c}}}, \bibinfo{journal}{Linear Algebra and its
  Applications} \textbf{\bibinfo{volume}{423}}, \bibinfo{pages}{155}
  (\bibinfo{year}{2007}).

\bibitem[{\citenamefont{Banerjee and Jost}(2007)}]{banerjee2007spectrum}
\bibinfo{author}{\bibfnamefont{A.}~\bibnamefont{Banerjee}} \bibnamefont{and}
  \bibinfo{author}{\bibfnamefont{J.}~\bibnamefont{Jost}},
  \bibinfo{journal}{arXiv preprint arXiv:0705.3772}  (\bibinfo{year}{2007}).

\bibitem[{\citenamefont{Wu}(2016)}]{wu2016graphs}
\bibinfo{author}{\bibfnamefont{C.~W.} \bibnamefont{Wu}},
  \bibinfo{journal}{Discrete Mathematics} \textbf{\bibinfo{volume}{339}},
  \bibinfo{pages}{1377} (\bibinfo{year}{2016}).

\bibitem[{\citenamefont{Hom and Johnson}(1991)}]{hom}
\bibinfo{author}{\bibfnamefont{R.~A.} \bibnamefont{Hom}} \bibnamefont{and}
  \bibinfo{author}{\bibfnamefont{C.~R.} \bibnamefont{Johnson}},
  \bibinfo{journal}{Cambridge UP, New York}  (\bibinfo{year}{1991}).

\bibitem[{\citenamefont{Horodecki et~al.}(2009)\citenamefont{Horodecki,
  Horodecki, Horodecki, and Horodecki}}]{horodecki2009quantum}
\bibinfo{author}{\bibfnamefont{R.}~\bibnamefont{Horodecki}},
  \bibinfo{author}{\bibfnamefont{P.}~\bibnamefont{Horodecki}},
  \bibinfo{author}{\bibfnamefont{M.}~\bibnamefont{Horodecki}},
  \bibnamefont{and}
  \bibinfo{author}{\bibfnamefont{K.}~\bibnamefont{Horodecki}},
  \bibinfo{journal}{Reviews of Modern Physics} \textbf{\bibinfo{volume}{81}},
  \bibinfo{pages}{865} (\bibinfo{year}{2009}).

\bibitem[{\citenamefont{G{\"u}hne and T{\'o}th}(2009)}]{guhne2009}
\bibinfo{author}{\bibfnamefont{O.}~\bibnamefont{G{\"u}hne}} \bibnamefont{and}
  \bibinfo{author}{\bibfnamefont{G.}~\bibnamefont{T{\'o}th}},
  \bibinfo{journal}{Physics Reports} \textbf{\bibinfo{volume}{474}},
  \bibinfo{pages}{1} (\bibinfo{year}{2009}).

\bibitem[{\citenamefont{Peres}(1996)}]{peres1996}
\bibinfo{author}{\bibfnamefont{A.}~\bibnamefont{Peres}},
  \bibinfo{journal}{Physical Review Letters} \textbf{\bibinfo{volume}{77}},
  \bibinfo{pages}{1413} (\bibinfo{year}{1996}).

\bibitem[{\citenamefont{Horodecki}(1997)}]{horodecki1997}
\bibinfo{author}{\bibfnamefont{P.}~\bibnamefont{Horodecki}},
  \bibinfo{journal}{arXiv preprint quant-ph/9703004}  (\bibinfo{year}{1997}).

\bibitem[{\citenamefont{McMahon}(2007)}]{mcmahon2007quantum}
\bibinfo{author}{\bibfnamefont{D.}~\bibnamefont{McMahon}},
  \emph{\bibinfo{title}{Quantum computing explained}} (\bibinfo{publisher}{John
  Wiley \& Sons}, \bibinfo{year}{2007}).

\bibitem[{\citenamefont{Garnerone et~al.}(2012)\citenamefont{Garnerone, Giorda,
  and Zanardi}}]{garnerone2012bipartite}
\bibinfo{author}{\bibfnamefont{S.}~\bibnamefont{Garnerone}},
  \bibinfo{author}{\bibfnamefont{P.}~\bibnamefont{Giorda}}, \bibnamefont{and}
  \bibinfo{author}{\bibfnamefont{P.}~\bibnamefont{Zanardi}},
  \bibinfo{journal}{New Journal of Physics} \textbf{\bibinfo{volume}{14}},
  \bibinfo{pages}{013011} (\bibinfo{year}{2012}).

\bibitem[{\citenamefont{Braunstein
  et~al.}(2006{\natexlab{b}})\citenamefont{Braunstein, Ghosh, Mansour,
  Severini, and Wilson}}]{braunstein2006some}
\bibinfo{author}{\bibfnamefont{S.~L.} \bibnamefont{Braunstein}},
  \bibinfo{author}{\bibfnamefont{S.}~\bibnamefont{Ghosh}},
  \bibinfo{author}{\bibfnamefont{T.}~\bibnamefont{Mansour}},
  \bibinfo{author}{\bibfnamefont{S.}~\bibnamefont{Severini}}, \bibnamefont{and}
  \bibinfo{author}{\bibfnamefont{R.~C.} \bibnamefont{Wilson}},
  \bibinfo{journal}{Physical Review A} \textbf{\bibinfo{volume}{73}},
  \bibinfo{pages}{012320} (\bibinfo{year}{2006}{\natexlab{b}}).

\bibitem[{\citenamefont{Rahiminia and Amini}(2008)}]{rahiminia2008separability}
\bibinfo{author}{\bibfnamefont{H.}~\bibnamefont{Rahiminia}} \bibnamefont{and}
  \bibinfo{author}{\bibfnamefont{M.}~\bibnamefont{Amini}},
  \bibinfo{journal}{Quantum Information \& Computation}
  \textbf{\bibinfo{volume}{8}}, \bibinfo{pages}{664} (\bibinfo{year}{2008}).

\bibitem[{\citenamefont{Hildebrand et~al.}(2008)\citenamefont{Hildebrand,
  Mancini, and Severini}}]{severini2008}
\bibinfo{author}{\bibfnamefont{R.}~\bibnamefont{Hildebrand}},
  \bibinfo{author}{\bibfnamefont{S.}~\bibnamefont{Mancini}}, \bibnamefont{and}
  \bibinfo{author}{\bibfnamefont{S.}~\bibnamefont{Severini}},
  \bibinfo{journal}{Mathematical Structures in Computer Science}
  \textbf{\bibinfo{volume}{18}}, \bibinfo{pages}{205} (\bibinfo{year}{2008}).

\bibitem[{\citenamefont{Wu}(2006)}]{wu2006}
\bibinfo{author}{\bibfnamefont{C.~W.} \bibnamefont{Wu}},
  \bibinfo{journal}{Physics Letters A} \textbf{\bibinfo{volume}{351}},
  \bibinfo{pages}{18} (\bibinfo{year}{2006}).

\bibitem[{\citenamefont{Wang and Wang}(2007)}]{wang2007tripartite}
\bibinfo{author}{\bibfnamefont{Z.}~\bibnamefont{Wang}} \bibnamefont{and}
  \bibinfo{author}{\bibfnamefont{Z.}~\bibnamefont{Wang}},
  \bibinfo{journal}{JOURNAL OF COMBINATORICS} \textbf{\bibinfo{volume}{14}},
  \bibinfo{pages}{R40} (\bibinfo{year}{2007}).

\bibitem[{\citenamefont{Xie et~al.}(2013)\citenamefont{Xie, Zhao, and
  Wang}}]{xie2013separability}
\bibinfo{author}{\bibfnamefont{C.}~\bibnamefont{Xie}},
  \bibinfo{author}{\bibfnamefont{H.}~\bibnamefont{Zhao}}, \bibnamefont{and}
  \bibinfo{author}{\bibfnamefont{Z.}~\bibnamefont{Wang}}, \bibinfo{journal}{the
  electronic journal of combinatorics} \textbf{\bibinfo{volume}{20}},
  \bibinfo{pages}{P21} (\bibinfo{year}{2013}).

\bibitem[{\citenamefont{Wu}(2009)}]{wu2009multipartite}
\bibinfo{author}{\bibfnamefont{C.~W.} \bibnamefont{Wu}}, \bibinfo{journal}{the
  electronic journal of combinatorics} \textbf{\bibinfo{volume}{16}},
  \bibinfo{pages}{R61} (\bibinfo{year}{2009}).

\bibitem[{\citenamefont{Wu}(2010)}]{wu2010graphs}
\bibinfo{author}{\bibfnamefont{C.~W.} \bibnamefont{Wu}},
  \bibinfo{journal}{Discrete Mathematics} \textbf{\bibinfo{volume}{310}},
  \bibinfo{pages}{2811} (\bibinfo{year}{2010}).

\bibitem[{\citenamefont{Hui and Jiao}(2013)}]{hui2013separability}
\bibinfo{author}{\bibfnamefont{Z.}~\bibnamefont{Hui}} \bibnamefont{and}
  \bibinfo{author}{\bibfnamefont{F.}~\bibnamefont{Jiao}},
  \bibinfo{journal}{Chinese Physics Letters} \textbf{\bibinfo{volume}{30}},
  \bibinfo{pages}{090303} (\bibinfo{year}{2013}).

\bibitem[{\citenamefont{Li et~al.}(2015)\citenamefont{Li, Chen, and
  Yang}}]{li2015quantum}
\bibinfo{author}{\bibfnamefont{J.-Q.} \bibnamefont{Li}},
  \bibinfo{author}{\bibfnamefont{X.-B.} \bibnamefont{Chen}}, \bibnamefont{and}
  \bibinfo{author}{\bibfnamefont{Y.-X.} \bibnamefont{Yang}},
  \bibinfo{journal}{Quantum Information Processing}
  \textbf{\bibinfo{volume}{14}}, \bibinfo{pages}{4691} (\bibinfo{year}{2015}).

\bibitem[{\citenamefont{Zanardi et~al.}(2004)\citenamefont{Zanardi, Lidar, and
  Lloyd}}]{zanardi}
\bibinfo{author}{\bibfnamefont{P.}~\bibnamefont{Zanardi}},
  \bibinfo{author}{\bibfnamefont{D.~A.} \bibnamefont{Lidar}}, \bibnamefont{and}
  \bibinfo{author}{\bibfnamefont{S.}~\bibnamefont{Lloyd}},
  \bibinfo{journal}{Physical Review Letters} \textbf{\bibinfo{volume}{92}},
  \bibinfo{pages}{060402} (\bibinfo{year}{2004}).

\end{thebibliography}
%\input{partial.bbl}

\end{document}